\documentclass{llncs}

\usepackage{comment}
\usepackage{amsmath}
\usepackage{amssymb}
\usepackage{hyperref}
\usepackage[table,dvipsnames]{xcolor}
\usepackage{tikz}
\usetikzlibrary{decorations.pathmorphing}
\usepackage{fixmath}
\usepackage{booktabs}
\usepackage{braket}
\usepackage{graphicx}
\usepackage{here}
\usepackage{algorithm, algorithmic}

\newcommand*{\etal}{\textit{et~al.}}

\newcommand*{\algo}[1]{\ensuremath{\mathsf{#1}}}

\newcommand*{\event}[1]{\ensuremath{\mathsf{#1}}}

\newcommand*{\card}[1]{\lvert #1 \rvert}

\newcommand{\relmiddle}[1]{\mathrel{}\middle#1\mathrel{}}
\newcommand{\mymiddle}{\relmiddle{|}}

\usepackage{xspace}

\newcommand*{\Mclaw}{\algo{Mclaw}}

\newcommand*{\Func}{\algo{Func}}
\newcommand*{\BBHT}{\algo{BBHT}}
\newcommand*{\E}{{\mathrm{\bf E}}}

\newcommand*{\Img}{\mathrm{Im}}

\newcommand*{\A}{\mathcal{A}}
\newcommand*{\B}{\mathcal{B}}

\newcommand*{\BHT}{\algo{BHT}}
\newcommand*{\HSX}{\algo{HSX}}
\newcommand*{\MTPS}{\algo{MTPS}}

\newcommand*{\good}{\event{good}}

\newcommand*{\pregood}{\event{pregood}}
\newcommand*{\equal}{\event{equal}}

\usepackage{color}
\newcommand{\hnote}[1]{#1}

\newcommand{\tnote}[1]{#1}

\newcommand{\hhnote}[1]{#1}

\renewcommand{\algorithmicensure}{\textbf{Output:}}

\def\equationautorefname~#1\null{eq.~(#1)\null}

\begin{document}

\pagestyle{plain}

\title{Improved Quantum Multicollision-Finding Algorithm}

\author{
Akinori Hosoyamada\inst{1,2} \and
Yu Sasaki\inst{1} \and
Seiichiro Tani\inst{3} \and
Keita Xagawa\inst{1}
}
\institute{
 {NTT Secure Platform Laboratories, NTT Corporation. \\
 3-9-11, Midori-cho, Musashino-shi, Tokyo 180-8585, Japan.}
 \and
 {Department of Information and Communication Engineering, Nagoya University, Furo-cho, Chikusa-ku, Nagoya 464-8603, Japan.}
 \and
 {NTT Communication Science Laboratories, NTT Corporation. \\
 3-1, Morinosato-Wakamiya, Atsugi-shi, Kanagawa 243-0198, Japan.}
 \email{\{hosoyamada.akinori,sasaki.yu,tani.seiichiro,xagawa.keita\}@lab.ntt.co.jp}
}

\maketitle

\begin{abstract}
The current paper improves the number of queries of the previous quantum multi-collision finding algorithms presented by Hosoyamada \etal~at Asiacrypt 2017.
Let an $l$-collision be a tuple of $l$ distinct inputs that result in the same output of a target function.
In cryptology, it is important to study how many queries are required to find $l$-collisions for random functions of which domains are larger than ranges.
The previous algorithm finds an $l$-collision for a random function by recursively calling the algorithm for finding $(l-1)$-collisions, and it achieves the average quantum query complexity of $O(N^{(3^{l-1}-1) / (2 \cdot 3^{l-1})})$, where $N$ is the range size of target functions.
The new algorithm removes the redundancy of the previous recursive algorithm so that different recursive calls can share a part of computations. The new algorithm finds an $l$-collision for random functions with the average quantum query complexity of $O(N^{(2^{l-1}-1) / (2^{l}-1)})$, which improves the previous bound for all $l\ge 3$ (the new and previous algorithms achieve the optimal bound for $l=2$). 
More generally, 
the new algorithm achieves the average quantum query complexity of $O\left(c^{3/2}_N N^{\frac{2^{l-1}-1}{ 2^{l}-1}}\right)$ for a random function $f\colon X\to Y$ such that
$|X| \geq l \cdot |Y| / c_N$ for any $1\le c_N \in  o(N^{\frac{1}{2^l - 1}})$.
With the same query complexity, it also finds a multiclaw for random functions, which is harder to find than a multicollision.

\bigskip
\textbf{Keywords}
post-quantum cryptography, quantum algorithm, multiclaw, multicollision
\end{abstract}


\section{Introduction}\label{sec:introduction}

Post-quantum cryptography has recently been discussed very actively in the cryptographic community. Quantum computers would completely break many classical public-key cryptosystems. In response, NIST is now conducting a standardization to select new public-key cryptosystems that resist attacks with quantum computers. Given this background, it is now important to investigate how quantum computers can impact on other cryptographic schemes including cryptographic hash functions. 

A multicollision for a function $f$ denotes multiple inputs to $f$ such that they are mapped to the same output value. In particular, an $l$-collision denotes a tuple of $l$ distinct inputs $x_1,x_2,\cdots,x_l$ such that $f(x_1) = f(x_2) = \cdots = f(x_l)$.

A multicollision is an important object in cryptography.
Lower bounds \tnote{on} the complexity of finding a multicollision \tnote{are} sometimes used to derive security bounds in the area of provable security (e.g., security bounds for the schemes based on the sponge construction~\cite{DBLP:conf/asiacrypt/JovanovicLM14}).
In a similar context, the complexity of finding a multicollision directly impacts on the best cryptanalysis against some constructions. Furthermore, multicollisions can be used as a proof-of-work for blockchains. In digital payment schemes, a coin must be a bit-string the validity of which can be easily checked but which is hard to produce. A micro-payment scheme, MicroMint~\cite{DBLP:conf/spw/RivestS96}, defines coins as $4$-collisions for a function. If $4$-collisions can be produced quickly, a malicious user can counterfeit coins.
Some recent works prove the security of schemes and protocols based on the assumption that there exist functions for which it is hard to find multicollisions~\cite{DBLP:conf/stoc/BitanskyKP18,DBLP:conf/eurocrypt/BermanDRV18,DBLP:conf/eurocrypt/KomargodskiNY18}.

Hosoyamada \etal~\cite{DBLP:conf/asiacrypt/HosoyamadaSX17} provided a survey of multicollision finding algorithms with quantum computers. They first showed that an $l$-collision can be produced with at most $O(N^{1/2})$ queries on average to the target random function with range size $N$ by iteratively applying the Grover search \cite{DBLP:conf/stoc/Grover96,boyer1998tight} $l$ times. They also reported that a combination of Zhandry's algorithm with $l=3$ \cite{DBLP:journals/qic/Zhandry15} and Belovs' algorithm \cite{DBLP:conf/focs/Belovs12} achieves $O(N^{10/21})$ for $l=3$, which is faster than the simple application of Grover's algorithm. Finally, Hosoyamada~\etal~presented their own algorithm that recursively applies the collision finding algorithm by Brassard, H{\o}yer, and Tapp~\cite{DBLP:conf/latin/BrassardHT98}. Their algorithm achieves the average query complexity of $O(N^{(3^{l-1}-1) / (2 \cdot 3^{l-1})})$ for every $l\ge 2$. For $l=3$ and $l=4$, the complexities are \tnote{$O(N^{4/9})$ and $O(N^{13/27})$}, respectively, and the algorithm works as follows.
\begin{itemize}
\item To search for 3-collisions, it first iterates the $O(N^{1/3})$-query quantum algorithm for finding a $2$-collision $O(N^{1/9})$ times. Then, it searches for the preimage of \tnote{any} one of the $O(N^{1/9})$ 2-collisions by using Grover's algorithm, which runs with $O(N^{4/9})$ queries.
\item To search for 4-collisions, it iterates the \tnote{$O(N^{4/9})$}-query quantum algorithm for finding a $3$-collision $O(N^{1/27})$ times. Then, it searches for the preimage of any one of the $O(N^{1/27})$ 3-collisions with $O(N^{13/27})$ queries.
\end{itemize}

As demonstrated above, the recursive algorithm by Hosoyamada~\etal~\cite{DBLP:conf/asiacrypt/HosoyamadaSX17} runs $(l-1)$-collision algorithm multiple times, but in each invocation, the algorithm starts from scratch. This fact motivates us to consider reusing the computations when we search for \tnote{multiple} $(l-1)$-collisions.

\subsubsection{Our Contributions.}
In this paper, we improve the quantum query complexity of the previous multicollision finding algorithm by removing the redundancy of \tnote{the algorithm}.
Consider the problem of finding an $l$-collision of a random function $f \colon X \rightarrow Y$, where $l\ge 2$ is an integer constant and $|Y|=N$.
In addition, suppose that there exists a parameter $c_N \geq 1$ such that $c_N = o(N^{\frac{1}{2^l - 1}})$ and $|X| \geq l \cdot |Y| / c_N$.
Then, the new algorithm achieves the average quantum query complexity of $O\left(c^{3/2}_N N^{\frac{2^{l-1}-1}{ 2^{l}-1}}\right)$.
In particular, if we can take $c_N$ as a constant, then our algorithm can find an $l$-collision of a random function with $O\left(N^{\frac{2^{l-1}-1}{ 2^{l}-1}}\right)$ queries on average, which improves the previous quantum query complexity $O\left(N^{\frac{3^{l-1}-1}{2 \cdot 3^{l-1}-1}}\right)$~\cite{DBLP:conf/asiacrypt/HosoyamadaSX17} and matches with the lower bound proved by Liu and Zhandry~\cite{DBLP:journals/iacr/LiuZ18}.

The complexities for small $l$'s are listed in~\autoref{tbl:numbers}.
A comparison between them can be found in~\autoref{fig:Known-Bounds-Rnd}.
Our algorithm finds a $2$-collision, $3$-collision, $4$-collision, and $5$-collision of SHA3-512 with $2^{170.7}$, $2^{219.4}$, $2^{238.9}$, and $2^{247.7}$ quantum queries, respectively, \tnote{up to a constant factor} (\autoref{tbl:numbers512}).

Moreover, our new algorithm finds multiclaws for random functions, which are harder to find than multicollisions:
An $l$-claw for functions $f_i \colon X_i \rightarrow Y$ for $1 \leq i \leq l$ is defined as a tuple $(x_1, \dots, x_l)\in X_1\times \cdots\times X_l$ such that $f_i(x_i) = f_j(x_j)$ for all $(i,j)$.
If there exists a parameter $c_N \geq 1$ such that $c_N = o(N^{\frac{1}{{2^l}-1}})$ and $|X_i| \geq |Y| / c_N$ for each $i$, our quantum algorithm finds an $l$-claw for random functions $f_i$'s with $O\left(c^{3/2}_N N^{\frac{2^{l-1}-1}{2^{l}-1}}\right)$ quantum queries on average.
In particular, if we can take $c_N$ as a constant, then our algorithm can find an $l$-claw with $O\left(N^{\frac{2^{l-1}-1}{ 2^{l}-1}}\right)$ quantum queries.

In this paper, we do not provide the analyses of other complexity measures such as time/space complexity and the depth of quantum circuits, but it is not difficult to show 
with analyses similar to those in Ref.~\cite{DBLP:conf/asiacrypt/HosoyamadaSX17}
that the space complexity and the circuit depth are the same order as the query complexity up to a polylogarithmic factor.

Hereafter, we only consider \emph{average} quantum query complexity over random functions
as the performance of algorithms unless stated otherwise.

\begin{table}[!htb]
\begin{center}
\caption{\tnote{Query} complexities of $l$-collision finding quantum algorithms. Each fraction denotes the logarithm of the number of queries to the base $N$.
The query complexity asymptotically approaches $1/2$ as $l$ increases.}
\label{tbl:numbers}
\renewcommand{\arraystretch}{2}
\begin{tabular}{c | c@{}c@{}c@{}c@{}c@{}c@{}c }
\toprule
$l$ & \makebox[1cm]{2}  & \makebox[1cm]{3} & \makebox[1cm]{4} & \makebox[1cm]{5} & \makebox[1cm]{6} & \makebox[1cm]{7} & \makebox[1cm]{8} \\
\midrule
\quad \cite{DBLP:conf/asiacrypt/HosoyamadaSX17} : {\large $\frac{3^{l-1}-1}{2 \cdot 3^{l-1}}$} & {\Large $\frac{1}{3}$}& {\Large $\frac{4}{9}$} & {\Large $\frac{13}{27}$} & {\Large $\frac{40}{81}$} & {\Large $\frac{121}{243}$} & {\Large $\frac{364}{729}$} & {\Large $\frac{1093}{2187}$}  \\
Ours : {\large $\frac{2^{l-1}-1}{2^{l}-1}$} & {\Large $\frac{1}{3}$}& {\Large $\frac{3}{7}$} & {\Large $\frac{7}{15}$} & {\Large $\frac{15}{31}$} & {\Large $\frac{31}{63}$} & {\Large $\frac{63}{127}$} & {\Large $\frac{127}{255}$}  \\ \bottomrule
\end{tabular}
\begin{tabular}{c | c@{}c@{}c@{}c@{}c@{}c@{}c }
\toprule
$l$ & \makebox[1cm]{2}  & \makebox[1cm]{3} & \makebox[1cm]{4} & \makebox[1cm]{5} & \makebox[1cm]{6} & \makebox[1cm]{7} & \makebox[1cm]{8} \\
\midrule
\quad \cite{DBLP:conf/asiacrypt/HosoyamadaSX17} : {\large $\frac{3^{l-1}-1}{2 \cdot 3^{l-1}}$} & {\small $0.3333..\ $}& {\small $0.4444..\ $} & {\small $0.4814..\ $} & {\small $0.4938..\ $} & {\small $0.4979..\ $} & {\small $0.4993..\ $} & {\small $0.4997..\ $}  \\
Ours : {\large $\frac{2^{l-1}-1}{2^{l}-1}$} & {\small $0.3333..\ $}& {\small $0.4285..\ $} & {\small $0.4666..\ $} & {\small $0.4838..\ $} & {\small $0.4920..\ $} & {\small $0.4960..\ $} & {\small $0.4980..\ $}  \\ \bottomrule
\end{tabular}
\end{center}
\end{table}

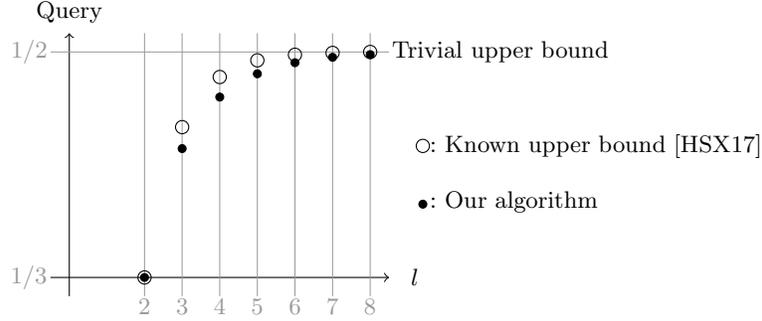
\begin{figure}[tbp]
\centering
\begin{tikzpicture}[scale=.5,
inner sep=1pt,
ours/.style={circle,draw=black,fill=black,thin,radius=2pt}]
\draw[->,thin] (-0.5,0) -- (8.5,0);
\draw[->,thin] (0,-1/2) -- (0,6.5);
\node[anchor=south] at (0,6.7) {Query};
\node[anchor=west] at (9,0) {$l$};
\node[gray!80,anchor=east] at (-0.5,0) {$1/3$};
\draw[gray!80,very thin] (-1/2,6) node[anchor=east] {$1/2$} -- (8.5,6); 
\node[anchor=west] at (8.5,6) {Trivial upper bound};
\foreach \x in {2,...,8} {\draw[gray!80,very thin] (\x,-1/2) node[anchor=north] {$\x$} -- (\x,6.5);}
\draw (2,36 * 1/3 - 12)			circle [radius=5pt];
\draw (3,36 * 4/9 - 12)			circle [radius=5pt];
\draw (4,36 * 13/27 - 12)			circle [radius=5pt];
\draw (5,36 * 40/81 - 12)			circle [radius=5pt];
\draw (6,36 * 121/243 - 12)		circle [radius=5pt];
\draw (7,36 * 364/729 - 12)		circle [radius=5pt];
\draw (8,6)		circle [radius=5pt];
\filldraw [ours] (2,36 * 1/3 - 12)			circle [radius=3pt];
\filldraw [ours] (3,36 * 3/7 - 12)			circle [radius=3pt];
\filldraw [ours] (4,36 * 7/15 - 12)			circle [radius=3pt];
\filldraw [ours] (5,36 * 15/31 - 12)			circle [radius=3pt];
\filldraw [ours] (6,36 * 31/63 - 12)		circle [radius=3pt];
\filldraw [ours] (7,36 * 63/127 - 12)		circle [radius=3pt];
\filldraw [ours] (8,36 * 127/255 - 12)		circle [radius=3pt];
\filldraw [ours] (9.4,1.95) circle [radius=3pt];
\node [anchor=west] at (9.5,2) {: Our algorithm};
\draw (9.4,3.5) circle [radius=5pt];
\node [anchor=west] at (9.5,3.5) {: Known upper bound~\cite{DBLP:conf/asiacrypt/HosoyamadaSX17}};
\end{tikzpicture}
\caption{Quantum query complexity for finding an $l$-collision. ``Query'' denotes the logarithm of  the number of queries to the base $N$.}
\label{fig:Known-Bounds-Rnd}
\end{figure}

\begin{table}[!htb]
\begin{center}
\caption{The number of queries required to find an $l$-collision of SHA3-512.
The numbers in the first row are obtained from the concrete bound given in~\cite[Thm.5.1]{DBLP:conf/asiacrypt/HosoyamadaSX17}, and those in the second row are obtained from the concrete bound given in \autoref{thm:main} with $k=2$.}
\label{tbl:numbers512}
\renewcommand{\arraystretch}{2}
\begin{tabular}{c | c@{}c@{}c@{}c@{}c@{}c@{}c}
\toprule
$l$ & \makebox[1cm]{2}  & \makebox[1cm]{3} & \makebox[1cm]{4} & \makebox[1cm]{5} \\
\midrule
\cite[Thm 5.1]{DBLP:conf/asiacrypt/HosoyamadaSX17} & {$2^{179}$}& {$2^{238}$} & {$2^{260}$} & {$2^{268}$}  \\
Ours,~\autoref{thm:main} & { $2^{181}$}& {$2^{230}$} & {$2^{250}$} & {$2^{259}$} \\ \bottomrule
\end{tabular}
\end{center}
\end{table}

\subsubsection{Paper Outline.}
The remaining of this paper is organized as follows.
\hnote{In \autoref{sec:prelim}, we describe notations, definitions and settings.
In \autoref{sec:previous}, we review previous works related to the multicollision-finding problem.
In \autoref{sec:alg}, we provide our new quantum algorithm and its complexity analysis.
In \autoref{sec:conclusion}, we conclude this paper.}

\subsubsection{Concurrent Work.}
Very recently, Liu and Zhandry~\cite{DBLP:journals/iacr/LiuZ18} showed that
 for every integer constant $l\ge 2$, $\Theta\left(N^{\frac{1}{2}(1-\frac{1}{2^l-1})}\right)$ quantum queries are both necessary and sufficient to find a $l$-collision with constant probability, for a random function. That is, they gave an improved upper bound and a new lower bound on the average case. 
The comparisons are summarized as follows:
\begin{itemize}
\item Liu and Zhandry consider the $l$-collision case that $|X| \geq l |Y|$, where $X$ is the domain and $Y$ is the range.
We treat the case that $|X| \geq \frac{l}{c_N} |Y|$ holds for any positive value $c_N\ge 1$ which is in $o(N^{\frac{1}{2^l-1}})$. We also consider the \emph{multiclaw} case.
\item Their exponent $\frac{1}{2}(1-\frac{1}{2^{l}-1})$ is the same as ours $\frac{2^{l-1}-1}{2^l-1}$.
\item They give the upper bound $O\left(N^{\frac{1}{2}(1-\frac{1}{2^l-1})}\right)$, while we give $O\left(c^{3/2}_N N^{\frac{1}{2}(1-\frac{1}{2^l-1})}\right)$. When $c_N$ is a constant, our bound matches their bound.
\item They give a lower bound, which matches with their and our upper bound. 
\end{itemize}
We finally note that our result on an improved $l$-collision finding algorithm
 for the case $|X| \geq l |Y|$ with query complexity $O\left(N^{\frac{1}{2}(1-\frac{1}{2^l-1})}\right)$
 is reported in the Rump Session of Asiacrypt 2017.

\section{Preliminaries}
\label{sec:prelim}
For a positive integer $M$, let $[M]$ denote the set $\{1,\dots,M\}$.
In this paper, $N$ denotes a positive integer.
We assume that $l$ is a positive integer constant.
We focus on reducing quantum \emph{query} complexities for finding multicollisions and multiclaws.
Unless otherwise noted, all sets are non-empty and finite.
For \tnote{sets} $X$ and $Y$, $\Func(X,Y)$ denotes the set of functions from $X$ to $Y$.
For each $f \in \Func(X,Y)$, we denote the set $\{f(x) \mid x \in X\}$ by $\mathrm{Im}(f)$.
For a \tnote{set} $X$, let $U(X)$ denote the uniform distribution over $X$.
For a distribution $\mathcal{D}$ on a set $X$, let $x \sim \mathcal{D}$ denote that $x$ is a random variable that takes \tnote{a value drawn from $X$ according to $\mathcal{D}$}.
When we say that an oracle of a function $f \colon X \rightarrow Y$ is available, we consider the situation that each elements of $X$ and $Y$ are encoded into suitable binary strings, and the oracle gate $O_f \colon \ket{x,z} \mapsto \ket{x,z \oplus f(x)}$ is available.

An \emph{$l$-collision} for a function $f\colon X \rightarrow Y$ is a tuple of elements $(x_1,\dots,x_l,y)$ in $X^\ell \times Y$ such that $f(x_i)=f(x_j)=y$ and $x_i \neq x_j$ for all $1 \leq i \neq j \leq l$.
An $l$-collision is simply called a \emph{collision} for $l=2$, and called a \emph{multicollision} for $l \geq 3$.
Moreover, an \emph{$l$-claw} for functions $f_i\colon X_i\rightarrow Y$ for $1 \leq i \leq l$ is a tuple $(x_1,\dots,x_l,y) \in X_1 \times \cdots \times X_l \times Y$ such that $f_1(x_1)= \cdots = f_l(x_l)=y$.
An $l$-claw is simply called a \emph{claw} for $l=2$, and called a \emph{multiclaw} for $l \geq 3$.

\tnote{The problems} of finding multicollisions or multiclaws are often studied in the \tnote{contexts} of both cryptography and quantum computation, but the problem \tnote{settings} of interest change depending on the \tnote{contexts}.
In the context of quantum computation, most problems are studied in the \emph{worst case}, and an algorithm is said to \tnote{(efficiently)} solve a problem only when \tnote{it does (efficiently)} for all functions.
On the other hand, most problems in cryptography are studied in the \emph{average case}, since randomness is one of the most crucial notions in cryptography.
In particular, we say that an algorithm \tnote{(efficiently)} solves a problem if it does so with a high probability on average over randomly chosen functions.

This paper focuses on the settings of interest in the context of cryptography.
Formally, our goal is to solve the following two problems.
\begin{problem}[Multicollision-finding problem, average case]\label{prob:mcoll}
Let $l\ge 2$ be a positive integer constant, and $X,Y$ denote \tnote{non-empty} finite sets.
Suppose that \tnote{a function} $F \colon X \rightarrow Y$ is chosen uniformly at random and \tnote{given} as a quantum oracle.
Then, find an $l$-collision for $F$.
\end{problem}
\begin{problem}[Multiclaw-finding problem, average case]\label{prob:mclaw}
Let $l\ge 2$ be a positive integer constant, and $X_1, \dots, X_l,Y$ denote \tnote{non-empty} finite sets.
Suppose \tnote{that} functions $f_i \colon X_i \rightarrow Y (1 \leq i \leq l)$ are chosen \tnote{independently and} uniformly at random, and \tnote{given} as quantum oracles.
Then, find an $l$-claw for $f_1,\dots,f_l$.
\end{problem}

Roughly speaking, \autoref{prob:mcoll} is easier to solve than \autoref{prob:mclaw}.
Suppose \tnote{that} $F \colon X \rightarrow Y$ is a function, and we want to find an $l$-collision for $F$.
Let $X_1,\dots,X_l$ be subsets of $X$ such that $X_i \cap X_j = \emptyset$ for $i \neq j$ and $\bigcup_i X_i = X$.
If $(x_1,\dots,x_l,y)$ is an $l$-claw for $F|_{X_1}, \dots, F|_{X_l}$, then it is obviously an $l$-collision for $F$.
In general, an algorithm for finding an \tnote{$l$-claw} can be converted into one for finding an $l$-collision.
To be precise, the following lemma holds.
\begin{lemma}\label{lem:clawtocoll}
Let $X,Y$ be non-empty \tnote{finite} sets, and $X_1,\dots,X_l$ be subsets of $X$ such that $X_i \cap X_j = \emptyset$ for $i \neq j$ and $\bigcup_i X_i = X$.
If there exists a quantum algorithm $\A$ that solves \autoref{prob:mclaw} for the sets $X_1,\dots,X_l,Y$ by making at most $q$ quantum queries with probability at least $p$, then there exists a quantum algorithm $\B$ that solves \autoref{prob:mcoll} for the sets $X,Y$ by making at most $q$ quantum queries with probability at least $p$.
\end{lemma}

How to measure the size of a problem also changes depending on which context we are in.
In the context of cryptography, the problem size is often regarded as the size of the range of functions in the problem rather than the size of the domains, since the domains of cryptographic functions such as hash functions are much larger than their ranges.
Hence, we regard the range size $|Y|$ as the size of \autoref{prob:mcoll} (and~\autoref{prob:mclaw}) when we analyze the complexity of quantum algorithms.

In the context of quantum computation, there exist previous works on problems related to ours~\cite{DBLP:conf/focs/Belovs12,DBLP:conf/focs/Ambainis04,DBLP:journals/tcs/Tani09,DBLP:conf/coco/BuhrmanDHHMSW01} (element distinctness problem, for example), but those works usually focus on the worst case complexity and regard the domain sizes of functions as the problem size.
In particular, there does not exist any previous work that studies multiclaw-finding problem for general $l$ in the average case, to the best of authors' knowledge.

\section{Previous Works}
\label{sec:previous}

\subsection{The Grover Search and Its Generalization}
As a main tool for developing quantum algorithms, we use the quantum database search algorithm that was originally developed by Grover~\cite{DBLP:conf/stoc/Grover96} and later generalized by Boyer, Brassard, H{\o}yer, and Tapp~\cite{boyer1998tight}.
Below we introduce the generalized version.
\begin{theorem}\label{thm:BBHT}
Let $X$ be a non-empty finite set and $f \colon  X \rightarrow \{0,1\}$ be a function such that $t / |X| < 17/81$, where $t = |f^{-1}(1)| $.
Then, there exists a quantum algorithm $\BBHT$ that finds $x$ such that $f(x)=1$ with an expected number of quantum queries to $f$ at most 
\[
\frac{4|X|}{\sqrt{(|X|-t)t}} \leq \frac{9}{2} \cdot \sqrt{\frac{|X|}{t}}.
\]
If $f^{-1}(1) = \emptyset$, then $\BBHT$ runs forever.
\end{theorem}
\autoref{thm:BBHT} implies that we can find $l$-collisions and $l$-claws for random functions with \tnote{$O(\sqrt{N})$} quantum queries, if the sizes of range(s) and domain(s) of function(s) are $\Theta(N)$:
Suppose that we are given random functions $f_i \colon  X_i \rightarrow Y$ for $1 \leq i \leq l$, where $|X_1|,\dots,|X_l|$, and $|Y|$ are all in $\Theta(N)$, and we want to find an $l$-claw for those functions.
Take an element $y \in Y$ randomly, and define $F_i \colon  X_i \rightarrow \{0,1\}$ for each $i$ by $F_i(x)=1$ if and only if $f_i(x)=y$.
Then, \tnote{effectively}, by applying \BBHT~to each $F_i$, we can find $x_i \in X_i$ such that $f_i(x_i)=y$ for each $i$ with $O(\sqrt{N})$ quantum queries  with a constant probability.
Similarly we can find an $l$-collision for a random function $F\colon [N] \rightarrow [N]$ with $O( \sqrt{N})$ quantum queries.
In particular, $O(\sqrt{N})$ is a trivial upper bound of \autoref{prob:mcoll} and \autoref{prob:mclaw}.

\subsection{The BHT Algorithm}
Brassard, H\o yer, and Tapp~\cite{DBLP:conf/latin/BrassardHT98} developed a quantum algorithm that finds $2$-claws (below we call it \BHT).\footnote{
As in our case, the BHT algorithm also focus on only quantum query complexity.
Although it runs in time $\tilde{O}(N^{1/3})$ on an idealized quantum computer, it requires $\tilde{O}(N^{1/3})$ qubits to store data in quantum memories.
Recently Chailloux et al.~\cite{DBLP:conf/asiacrypt/ChaillouxNS17} has developed a quantum $2$-collision finding algorithm that runs in time $\tilde{O}(N^{2/5})$, which is polynomially slower than the BHT algorithm but requires only $O(\log N)$ quantum memories.
}
\tnote{\BHT~finds} a claw for two \tnote{one-to-one} functions $f_1\colon X_1 \rightarrow Y$ and $f_2 \colon  X_2 \rightarrow Y$ \tnote{as sketched in the following}.
For simplicity, here we assume $|X_1| = |X_2| = |Y|=N$.
Under this setting, \BHT~finds a $2$-claw with $O(N^{1/3})$ quantum queries.

\paragraph{\tnote{Rough Sketch of \BHT}}:
\begin{description}
\item[1. Construction of a list $L$.] Take a subset $S \subset X_1$ of size $N^{1/3}$ arbitrarily.
For each $x \in S$, \tnote{compute} the value $f_1(x)$ by making a query and store the pair $(x,f_1(x))$ in a list $L$.
\item[2. Extension to a claw.] Define a function $F_L \colon  X_2 \rightarrow \{0,1\}$ by $F_L(x')=1$ if and only if the value $f_2(x') \in Y$ appears in the list $L$ (i.e., there exists $x_1 \in S$ such that $f_2(x')=f_1(x_1)$).
Apply $\BBHT$~to $F_L$ and find $x_2 \in X_2$ such that $f_2(x_2)$ appears in $L$.
\item[3. Finalization.] Find $(x_1,f_1(x_1)) \in L $ such that $f_1(x_1) = f_2(x_2)$, and then output $(x_1,x_2)$.
\end{description}

\paragraph{Quantum query complexity.}
\BHT~finds a claw with $O(N^{1/3})$ quantum queries.
In the first step, the list $L$ is constructed by making $N^{1/3}$ queries to $f_1$.
In the second step, since $|F^{-1}_L(1)| = |f^{-1}_2(f_1(S))|$ is equal to $N^{1/3}$, $\BBHT$ finds $x_2$ with $O(\sqrt{N/N^{1/3} }) = O(N^{1/3})$ quantum queries to $f_2$ (note that we can evaluate $F_L$ by making one query to $f_2$).
The third step does not require queries.
Therefore \BHT~finds a collision by making $O(N^{1/3})$ quantum queries in total in the worst case.

\subsubsection{Extension to a collision-finding algorithm.}
It is not difficult to show that \BHT~works for random functions.
Thus, $\BHT$~can be extended to the quantum collision-finding algorithm
as mentioned in~\autoref{sec:prelim}.
Suppose we want to find a ($2$-)collision for a random function $F \colon  X \rightarrow Y$.
Here we assume $|X|=2N$ and $|Y|=N$ for simplicity.
Now, choose a subset $X_1 \subset X$ of size $N$ arbitrarily and let $X_2 \colon= X \setminus X_1$.
Then we can find a collision for $F$ by applying the $\BHT$ algorithm introduced above to the functions $F|_{X_1}$ and $F|_{X_2}$, since a claw for them becomes a collision for $F$.

\subsection{The HSX Algorithm}
Next, we introduce a quantum algorithm for finding multicollisions \tnote{that was} developed by Hosoyamada, Sasaki, and Xagawa~\cite{DBLP:conf/asiacrypt/HosoyamadaSX17} (the algorithm is designed to find only multicollisions, and cannot find multiclaws).
Below we call their algorithm \HSX.

The main idea of \HSX~is to apply the strategy of \BHT~recursively:
To find an $l$-collision, \HSX~calls itself recursively to find many $(l-1)$-collisions, and then extend any one of those $(l-1)$-collisions to an $l$-collision by applying \BBHT.

\paragraph{\tnote{Rough Sketch of \HSX}}:
In what follows, $N$ denotes $|Y|$.
Let us denote $\HSX(l)$ by the HSX algorithm \tnote{for finding} $l$-collisions.
$\HSX(l)$ finds an $l$-collision for a random function $f\colon X \rightarrow Y$ with $|X| \geq l \cdot |Y|$ as follows.
\begin{description}
\item[Recursive call to construct a list $L_{l-1}$.] Apply $\HSX(l-1)$ to $f$ $N^{1/3^{l-1}}$ times to obtain $N^{1/3^{l-1}}$ many $(l-1)$-collisions. Store those $(l-1)$-collisions in a list $L_{l-1}$.
\item [Extension to an $l$-collision.] Define $F_{l-1} \colon  X \rightarrow \{0,1\}$ by $F_{l-1}(x')=1$ if and  only if there exists an $(l-1)$-collision $(x_1,\dots,x_{l-1},y) \in L_{l-1}$ such that $(x_1,\dots,x_{l-1},x',y)$ forms an $l$-collision for $f$, i.e., $f(x')=y$ and $x' \neq x_i$ for $1 \leq i \leq l-1$.
Apply $\BBHT$ to $F_{l-1}$ to find $x_{l} \in X$ such that $F_{l-1}(x_l)=1$.
\item[Finalization.] Find $(x_1,\dots,x_{l-1},y) \in L_{l-1}$ such that $F_{l-1}(x_l)=y$.
Output $(x_1,\dots,x_{l-1},x_l,y)$. 
\end{description}

\paragraph{Quantum query complexity.}
\HSX~finds a $l$-collision with $O(N^{(3^{l-1}-1)/2 \cdot 3^{l-1}})$ quantum queries on average, which can be shown by induction as follows.
For $2$-collisions, $\HSX(2)$ matches the \BHT~algorithm. 
For general $l \geq 3$, suppose that $\HSX(l-1)$ finds an $(l-1)$-collision with $O(N^{(3^{l-2}-1)/2 \cdot 3^{l-2}})$ quantum queries on average.
In its first step, $\HSX(l)$ makes $N^{1/3^{l-1}} \cdot O(N^{(3^{l-2}-1)/2 \cdot 3^{l-2}}) = O(N^{(3^{l-1}-1)/2 \cdot 3^{l-1}})$ quantum queries.
Moreover, in its second step, $\HSX(l)$ makes $O( \sqrt{N / N^{(3^{l-2}-1)/2 \cdot 3^{l-2}} }) = O(N^{(3^{l-1}-1)/2 \cdot 3^{l-1}})$ quantum queries by using \BBHT.
The third step does not make quantum queries.
Therefore it follows that $\HSX(l)$ makes $O(N^{(3^{l-1}-1)/2 \cdot 3^{l-1}})$ quantum queries in total.
\section{New Quantum Algorithm \Mclaw}
\label{sec:alg}
This section gives our new quantum algorithm \Mclaw~that finds an $l$-claw with $O(c^{3/2}_N N^{(2^{l-1}-1)/(2^l-1)})$ quantum queries for random functions $f_i \colon X_i \rightarrow Y$ for $1 \leq i \leq l$, where $\card{Y}=N$ and there exists a real value $c_N$ 
with $1\le c_N \in o(N^{\frac{1}{2^l -1}})$ such that $\frac{N}{c_N} \leq |X_i|$ holds for all $i$.
Roughly speaking, this means that, an $l$-collision for a random function $f \colon X \rightarrow Y$, where $\card{Y}=N$ and $\card{X}\ge l\cdot N$, can be found with $O(N^{(2^{l-1}-1)/(2^l-1)})$ quantum queries, which improves the previous result~\cite{DBLP:conf/asiacrypt/HosoyamadaSX17} (see~\autoref{sec:prelim}).

Our algorithm assumes that $\card{X_1},\dots,\card{X_l}$ are less than or equal to $\card{Y}$.
However, it can also be applied to the functions of interest in the context of cryptography, i.e., the functions of which domains are much larger than ranges, by restricting the domains of them to suitable subsets.

The main idea of our new algorithm is to improve $\HSX$  by getting rid of its redundancy:
To find an $l$-collision, $\HSX$ recursively calls itself to find many $(l-1)$-collisions.
Once $\HSX$ finds an $(l-1)$-collision $\gamma=(x_1,\dots,x_{l-1},y)$, it stores $\gamma$ in a list $L_{l-1}$, \emph{discards all the data that was used to find $\gamma$}, and then start to search for another $(l-1)$-collision $\gamma'$.
It is inefficient to discard data every time an $(l-1)$-collision is found, and our new algorithm $\Mclaw$ reduces the number of \tnote{quantum queries} by reusing those data.
We note that our algorithm $\Mclaw$ can solve the multiclaw-finding problem as well as the multi-collision finding problem.

We begin with describing our algorithm in an intuitive manner, and then give its formal description.

\subsection{Intuitive Description and Complexity Analysis}
We explain the idea of how to develop the $\BHT$ algorithm,
 how to develop a quantum algorithm to find $3$-claws from $\BHT$,
 and how to extend it further to \tnote{the case of finding an $l$-claw for any $l$}.
\subsubsection{How to develop the BHT algorithm.}
Here we review how the \BHT~algorithm is developed.
Let $f_1 \colon X_1 \rightarrow Y$ and $f_2 \colon X_2 \rightarrow Y$ be \tnote{one-to-one} functions.
The goal of the \BHT~algorithm is to find a ($2$-)claw for $f_1$ and $f_2$ with \tnote{$O(N^{1/3})$} quantum queries.
For simplicity, below we assume that $\card{X_1} = \card{X_2} = \card{Y} = N$ holds.
Let $t_1$ be a parameter that defines the size of a list of $1$-claws for $f_1$. It will be set as $t_1=N^{1/3}$. 

First, collect $t_1$ many $1$-claws for $f_1$ and store them in a list $L_1$.
This first step makes $t_1$ queries.
Second, extend one of $1$-claws in $L_1$ to a $2$-claw for $f_1$ and $f_2$, by using $\BBHT$, and output the obtained $2$-claw.
Since \BBHT~makes \tnote{$O(\sqrt{N/t_1})$} queries to make a $2$-claw from $L_1$, this second step makes \tnote{$O(\sqrt{N/t_1})$} queries (see \autoref{thm:BBHT}).
Overall, the above algorithm makes $q_2(t_1) = t_1 + \sqrt{N/t_1}$ quantum queries \tnote{up to a constant factor}.
The function $q_2(t_1)$ takes its minimum value $2 \cdot N^{1/3}$ when $t_1 = N^{1/3}$.
By setting $t_1 = N^{1/3}$, the BHT algorithm is obtained.

\subsubsection{From \BHT~to a $3$-claw-finding algorithm.}
Next, we show how the above strategy to develop the $\BHT$ algorithm can be extended to develop a $3$-claw-finding algorithm.
Let $f_i \colon X_i \rightarrow Y$ be \tnote{one-to-one} functions for $1 \leq i \leq 3$.
Our goal here is to find a $3$-claw for $f_1$, $f_2$, and $f_3$ with $O(N^{3/7})$ quantum queries.
For simplicity, below we assume $|X_1| = |X_2| = |X_3|= |Y|=N$.
Let $t_1,t_2$ be parameters that define the number of $1$-claws for $f_1$ and that of $2$-claws for $f_1$ and $f_2$, respectively. (They will be fixed later.)

First, collect $t_1$ many $1$-claws for $f_1$ and store them in a list $L_1$.
This first step makes $t_1$ queries.
Second, extend $1$-claws in $L_1$ to $t_2$ many $2$-claws for $f_1$ and $f_2$ by using \BBHT, and store them in a list $L_2$.
Here we do not discard the list $L_1$ until we construct the list $L_2$ of size $t_2$, \tnote{while} the $\HSX$ algorithm does.
Since \BBHT~makes \tnote{$O(\sqrt{N/t_1})$} queries to make a $2$-claw from $L_1$, this second step makes $t_2 \cdot O(\sqrt{N/t_1})$ queries if $t_2=o(t_1)$ (see~\autoref{thm:BBHT}).
Finally, extend one of $2$-claws in $L_2$ to a $3$-claw for $f_1$, $f_2$, and $f_3$ by using \BBHT, and output the obtained $3$-claw.
This final step makes \tnote{$O(\sqrt{N/t_2})$} queries.
Overall, the above algorithm makes $q_3(t_1,t_2) = t_1 + t_2 \cdot \sqrt{N/t_1} +  \sqrt{N/t_2}$ quantum queries \tnote{up to a constant factor}.
The function $q_3(t_1,t_2)$ takes its minimum value $3 \cdot N^{3/7}$ when $t_1 = t_2 \cdot \sqrt{N/t_1} =  \sqrt{N/t_2}$, which is equivalent to $t_1 = N^{3/7}$ and $t_2 = N^{1/7}$.
By setting $t_1=N^{3/7}$ and $t_2=N^{1/7}$,
we can obtain a $3$-claw finding algorithm with \tnote{$O(N^{3/7})$} quantum queries.
\subsubsection{$l$-claw-finding algorithm for general $l$.}
Generalizing the above idea to find a $3$-claw, we can find an $l$-claw for general $l$ as follows.
Let $f_i \colon X_i \rightarrow Y$ be \tnote{one-to-one} functions for $1 \leq i \leq l$.
Our goal here is to find an $l$-claw for $f_1,\dots,f_l$.
For simplicity, below we assume that $|X_1| = \cdots = |X_l|= |Y|=N$ holds.
Let $t_1,\dots,t_{l-1}$ be parameters with $t_i=o(t_{i-1})$ for $i=2,.\dots, l$.

First, collect $t_1$ many $1$-claws for $f_1$ and store them in a list $L_1$.
This first step makes $t_1$ queries.
In the $i$-th step for $2 \leq i \leq l-1$, extend  $t_i$ many $(i-1)$-claws in \tnote{$L_{i-1}$} to $t_i$ many $i$-claws for $f_1,\dots,f_i$ by using \BBHT, and store them in a list $L_i$.
Here we do not discard the list $L_{i-1}$ until we construct the list $L_i$ of size $t_i$.
Since \BBHT~makes \tnote{$O(\sqrt{N/t_{i-1}})$} queries to make an $i$-claw from $L_{i-1}$, the $i$-th step makes \tnote{$t_i \cdot O(\sqrt{N/t_{i-1}})$} queries.
Finally, extend one of $(l-1)$-claws in $L_{l-1}$ to an $l$-claw for $f_1,\dots,f_l$ by using \BBHT, and output the obtained $l$-claw.
This final step makes \tnote{$O(\sqrt{N/t_{l-1}})$} queries.
Overall, this algorithm makes $q_l(t_1,\dots,t_{l-1}) = t_1 + t_2 \cdot \sqrt{N/t_1} + \cdots + t_{l-1}\cdot\sqrt{N / t_{l-2}} + \sqrt{N/t_{l-1}}$ quantum queries \tnote{up to a constant factor}.
The function $q_l(t_1,\dots,t_{l-1})$ takes its minimum value $l \cdot N^{(2^{l-1}-1)/(2^l-1)}$ when $t_1 = t_2 \cdot \sqrt{N/t_1} = \cdots = t_{l-1}\cdot\sqrt{N / t_{l-2}} = \sqrt{N/t_{l-1}}$, which is equivalent to $t_i = N^{(2^{l-i}-1)/(2^l-1)}$.
By setting $t_i = N^{(2^{l-i}-1)/(2^l-1)}$,
we can find an $l$-claw with \tnote{$O(N^{(2^{l-1}-1)/(2^l-1)})$} quantum queries.
Our new quantum algorithm \Mclaw~is developed based on the above strategy \tnote{for random functions}.

\subsection{Formal Description}
\hhnote{Here we formally describe our quantum multiclaw-finding algorithm $\Mclaw$.}
\tnote{A formal} complexity analysis of \Mclaw~is given in the next subsection, and 
this subsection only \tnote{describes how the algorithm works}.

Let $N$ be a sufficiently large integer and suppose that $|Y|=N$ holds.
Below we assume that $\card{X_i} \leq \card{Y}$ holds for all $i$.
This is a reasonable assumption since, if there is an algorithm that solves~\autoref{prob:mclaw} in the case that $\card{X_i} \leq \card{Y}$ holds for all $i$, then we can also solve the problem in other cases:
If $\card{X_i} > \card{Y}$ holds for some $i$, take a subset $S_i \subset X_i$ such that $\card{S_i} = \card{Y}$ and find an $l$-claw for $f_1,\dots,f_{i-1},f_i|_{S_i}, f_{i+1},\dots, f_l$.
Then the $l$-claw is also an $l$-claw for $f_1, \dots, f_l$.

\hhnote{
Here we introduce a corollary that follows from~\autoref{thm:BBHT}.
\begin{corollary}\label{cor:MTPS}
Let $X,Y$ be non-empty finite sets, $f \colon X \rightarrow Y$ be a function, and $L' \subset Y$ be a non-empty subset.
Then there exists a quantum algorithm $\MTPS$ that finds $x$ such that $f(x) \in L'$ with an expected number of quantum queries to $f$ at most $9\sqrt{5\card{X}/\card{f^{-1}(L')}}$.
\end{corollary}
Let $F_{L'} \colon \{1,\dots,5\} \times X \rightarrow \{0,1\}$ be the boolean function defined by $F_{L'}(\alpha,x) = 1$ if and only if $\alpha=1$ and $f(x) \in L'$.
A quantum circuit that computes $F_{L'}$ can be implemented with two oracle calls to $f$.
Then, run $\BBHT$ on $F_{L'}$.
Since $\card{\{1,\dots,5\} \times X} = 5\card{X}$ and $\card{F_{L'}^{-1}(1)} \leq \card{X} \leq 17/81 \cdot \card{\{1,\dots,5\} \times X}$ always hold, we can show that the corollary follows from~\autoref{thm:BBHT}. 
}

Our algorithm is parametrized by a positive integer $k\ge 2$, and we denote the algorithm for the parameter $k$ by $\Mclaw_k$.
$\Mclaw_k$ can be applied in the situation that there exists a parameter $c_N \geq 1$ such that $c_N$ is in $o(N^{\frac{1}{2^l - 1}})$ and $|X_i| \geq |Y| / c_N$ holds for each $i$.
We impose an upper limit on the number of queries \tnote{that $\Mclaw_k$ is allowed to make}:
We design $\Mclaw_k$ in such a way that it immediately stops and aborts if the number of queries made reaches the limit specified by the parameter $\mathsf{Qlimit}_k := k \cdot 169 l c^{3/2}_N \cdot  N^{\frac{2^{l-1}-1}{2^{l}-1}}$.
The upper limit $\mathsf{Qlimit}_k$ is necessary to prevent the algorithm from running forever, and to make the expected value of the number of queries converge.
We also define the parameters controlling the sizes of the lists:
\begin{equation}
N_i \colon=
\begin{cases}
\frac{N}{4c_N} & (i=0), \\
N^{\frac{2^{l-i}-1}{2^l-1}} & (i \geq 1).
\end{cases}
\end{equation}
For ease of notation, we define $L_0$ and $L_0'$ as $L_0 = L'_0 = Y$.
Then, $\Mclaw_k$ is described as in Algorithm~\ref{alg:Mclaw}.

\begin{algorithm}
\caption{$\Mclaw_k$}
\label{alg:Mclaw}
\begin{algorithmic}
\REQUIRE Randomly chosen functions $f_1,\dots, f_l$ ($f_i \colon X_i \rightarrow Y$ and $\card{X_i} \leq \card{Y})$).
\ENSURE An $l$-claw for $f_1,\dots,f_l$ or $\perp$.
\renewcommand{\algorithmicensure}{\textbf{Stop condition:}}
\ENSURE If the number of queries reaches $\mathsf{Qlimit}_k$, stop and output $\perp$.
\STATE $L_1, \dots, L_{l} \gets \emptyset$, $L'_1, \dots, L'_{l} \gets \emptyset$.
\FOR{$i=1$ to $l$}
    \FOR{$j=1$ to $\left\lceil 4c_N \cdot N_i \right\rceil$}
        \IF{$i=1$}
            \STATE Take $x_j \in X_1$ that does not appear in $L_1$, $y \gets f_1(x_j)$. \hfill //$1$ query is made
        \ELSE
            \STATE Find $x_j \in X_i$ whose image $y \colon= f_i(x_j)$ is in \tnote{$L_{i-1}'$}
            by running $\MTPS$ on $f_i$ and $L'_{i-1}$. \quad \hfill //multiple queries are made
        \ENDIF
        \STATE $L_i \gets L_i \cup \{(x^{(1)}, \dots, x^{(i-1)},x_j,y)\}$, $L'_i \gets L'_i \cup \{y\}$.
        \STATE $L_{i-1} \gets L_{i-1} \setminus \{(x^{(1)}, \dots, x^{(i-1)},y)\}$, $L'_i \gets L'_{i-1} \setminus \{y\}$.
    \ENDFOR
\ENDFOR
\STATE Return an element $(x^{(1)},\dots,x^{(l)};y) \in L_l$ as an output.
\end{algorithmic}
\end{algorithm}

\subsection{Formal Complexity Analysis}
This section gives a formal complexity analysis of $\Mclaw_k$.
The goal of this section is to show the following theorem.
\begin{theorem}\label{thm:main}
Assume that there exists a parameter $c_N \geq 1$ such that $c_N$ is in $o(N^{\frac{1}{2^l-1}})$ and $\card{X_i} \geq \frac{1}{c_N}\card{Y}$ holds for each $i$.
If $\card{Y}=N$ is sufficiently large, $\Mclaw_k$ finds an $l$-claw with a probability at least
\begin{equation}
1 -\frac{1}{k} - \frac{2l}{N}  - l \cdot \exp\left(- \frac{1}{15} \cdot \frac{N^{\frac{1}{2^l - 1}}}{c_N}\right), \label{eq:MclawkProbLower}
\end{equation}
by making at most
\begin{equation}
\mathsf{Qlimit}_k = k \cdot  169 l c^{3/2}_N \cdot  N^{\frac{2^{l-1}-1}{2^{l}-1}} 
\end{equation}
quantum queries, where $k$ is any positive integer $2$ or more.
\end{theorem}
This theorem shows that, for each integer $k \geq 2$, $\Mclaw_k$ finds an $l$-claw with a constant probability by making $O\left(c^{3/2}_NN^{\frac{2^{l-1}-1}{2^{l}-1}}\right)$ queries.

For later use, we show the following lemma.
\begin{lemma}\label{lem:BallBin}
Let $X,Y$ be non-empty finite sets such that $\card{X} \leq \card{Y}$.
Suppose that a function $f \colon X \rightarrow Y$ is chosen uniformly at random.
Then
\begin{equation}
\Pr_{f \sim U(\Func(X,Y))} \left[ \card{{\rm Im}(f)} \geq  \frac{|X|}{2} - \sqrt{|X|\ln |Y|/2} \right] \geq 1 - \frac{2}{|Y|}
\end{equation}
holds.
\end{lemma}
\begin{proof}
Note that, for each $x \in X$, $f(x)$ is the random variable that takes value in $Y$.
Moreover, $\{f(x)\}_{x \in X}$ is the set of independent random variables.
Let us define a function $\Phi \colon Y^{\times \card{X}} \rightarrow \mathbb{N}$ by $\Phi \left(y_1,\dots,y_{\card{X}}\right) = \left| Y \setminus \{y_i\}_{1 \leq i \leq \card{X}}\right|$.
Then $\Phi$ is $1$-Lipschitz, i.e., 
\begin{equation}
\left| \Phi(y_1,\dots,y_{i-1},y_i,y_{i+1},\dots,y_{\card{X}}) - \Phi(y_1,\dots,y_{i-1},y'_i,y_{i+1},\dots,y_{\card{X}}) \right| \leq 1
\end{equation}
holds for arbitrary choice of elements $y_1,\dots,y_{\card{X}}$, and $y'_i$ in $Y$.
Now we apply the following theorem to $\Phi$.
\begin{theorem}[McDiarmid's Inequality (Theorem 13.7 in \cite{mitzenmacher2017probability})]
Let $M$ be a positive integer, and $\Phi \colon Y^{\times M} \colon \rightarrow \mathbb{N}$ be a $1$-Lipschitz function.
Let $\{y_i\}_{1 \leq i \leq M}$ be the set of independent random variables that take values in $Y$.
Let $\mu$ denote the expectation value $\E_{y_1,\dots,y_M}\left[ \Phi(y_1,\dots,y_M)  \right]$.
Then
\begin{equation}
\Pr_{y_1,\dots,y_M}\left[ \Phi(y_1,\dots,y_M) \geq \mu + \lambda \right] \leq 2e^{-2\lambda^2/M}
\end{equation}
holds.
\end{theorem}
Apply the above theorem with $M = \card{X}$, $\lambda = \sqrt{|X|\ln|Y| / 2}$, and $y_x = f(x)$ for each $x \in X$ (here we identify $X$ with the set $\{1,\dots,\card{X}\}$).
Then, since $\E \left[ \Phi(y_1,\dots,y_M) \right] = |Y|\left( 1- 1/\card{Y}\right)^{\card{X}}$ holds, we have that
\begin{align*}
\Pr_{f \sim U(\Func(X,Y))}\left[ \Phi(y_1,\dots,y_M) \geq \card{Y}\left( 1- 1/\card{Y}\right)^{\card{X}} + \sqrt{|X|\ln|Y| / 2} \right] \leq \frac{2}{\card{Y}}.
\end{align*}
In addition, it follows that
\begin{align}
\card{Y}\left( 1- 1/\card{Y}\right)^{\card{X}} &\leq \card{Y}e^{-\card{X}/\card{Y}} \leq \card{Y} \left(1 - \frac{|X|}{|Y|} + \frac{1}{2}\left(\frac{|X|}{|Y|}\right)^2 \right) \nonumber \\
&= |Y| - |X| \left( 1 - \frac{1}{2}\frac{|X|}{|Y|}\right) \leq |Y| - \frac{\card{X}}{2},
\end{align}
where we used the assumption that $\card{X} \leq \card{Y}$ for the last inequality.
Since $\Phi(y_1,\dots,y_M) = \left|Y \setminus {\rm Im}(f) \right|$ and $\left|{\rm Im}(f)\right| = |Y| -  \left|Y \setminus {\rm Im}(f) \right|$ hold, it follows that 
$\left|{\rm Im}(f)\right|$ 
is at least
\begin{align}
 |Y|  - \left(\card{Y} - \frac{\card{X}}{2} + \sqrt{\card{X} \ln \card{Y} / 2} \right) = \frac{\card{X}}{2} -  \sqrt{\card{X} \ln \card{Y} / 2}
\end{align}
with a probability at least $1 - \frac{2}{\card{Y}}$, which completes the proof.
\qed
\end{proof}

\begin{proof}[of~\autoref{thm:main}]
We show that~\autoref{eq:MclawkProbLower} holds.
Let us define $\good^{(i)}$ to be the event that
\begin{equation}
\card{\Img(f_i) \cap L_{i-1}'} \geq N_{i-1}
\end{equation}
holds just before $\Mclaw_k$ starts to construct $i$-claws.
(Intuitively, under the condition that $\good^{(i)}$ occurs, the number of queries does not become too large.)
We show the following claim.
\begin{claim}
For sufficiently large $N$,
\begin{equation}
\Pr\left[\good^{(i)}\right] \geq 1 - \frac{2}{N}  - \exp\left(- \frac{1}{15} \cdot \frac{N_{i-1}}{c_N}\right).
\end{equation}
holds.
\end{claim}
\begin{proof}
In this proof we consider the situation that $\Mclaw_k$ has finished to make $L_{i-1}$ and before starting to make $i$-claws.
In particular, we assume that $\card{L_{i-1}}=\card{L'_{i-1}}=\left\lceil4c_N N_{i-1}\right\rceil$.

Let $\pregood^{(i)}$ be the event that $\card{{\rm Im}(f_i)} \geq \left\lceil N/3c_N \right\rceil$ holds.
Since $c_N$ is in $o(N^{\frac{1}{2^l-1}})$, we have that $\frac{\card{X_i}}{2} - \sqrt{\card{X_i}\ln\card{Y}/2} \geq \left\lceil \frac{N}{3c_N} \right\rceil$ holds for sufficiently large $N$.
Hence
\begin{equation}
\Pr\left[ \pregood^{(i)} \right] \geq 1 - \frac{2}{\card{Y}}
\end{equation}
follows from~\autoref{lem:BallBin}.

Let us identify $X$ and $Y$ with the sets $\{1,\dots,\card{X}\}$ and $\{1,\dots,\card{Y}\}$, respectively.
Let $B_j$ be the $j$-th element in ${\rm Im}(f)$.
Let $\chi_j$ be the indicator variable that is defined by $\chi_j = 1$ if and only if $B_j \in L'_{i-1}$, and define a random variable $\chi$ by $\chi\colon= \sum_j \chi_j$.
Then $\chi$ follows the hypergeometric distribution.
We use the following theorem as a fact.
\begin{theorem}[Theorem 1 in~\cite{hush2005concentration}]
Let $K=K(n_1,n,m)$ denote the hypergeometric random variable describing the process of counting how many defectives are selected when $n_1$ items are randomly selected without replacement from $n$ items among which there are $m$ defective ones. Let $\lambda \geq 2$.
Then
\begin{equation}
\Pr\left[K - \E[K] < -\lambda \right] < e^{-2\alpha_{n_1,n,m}(\lambda^2-1)}
\end{equation}
holds, where
\begin{equation}
\alpha_{n_1,m,n} = \max\left( \left( \frac{1}{n_1+1} + \frac{1}{n-n_1+1} \right), \left(\frac{1}{m+1} + \frac{1}{n-m+1} \right)\right).
\end{equation}
\end{theorem}
Apply the above theorem with $n_1 = \left\lceil N / 3c_N\right\rceil$, $n=N$, and $m = |L'_{i-1}| = \left\lceil 4c_N N_{i-1} \right\rceil$, for the random variable $\chi$ under the condition that $\card{{\rm Im}(f_i)} = \left\lceil N/3c_N\right\rceil$ holds.
Let $\equal$ denote the event that $\card{{\rm Im}(f_i)} = \left\lceil N/3c_N \right\rceil$ holds.
Then $\E\left[\chi\middle|\equal\right]= \frac{n_1m}{n} \geq \frac{4}{3}{N_{i-1}}$ holds, and we have that 
\begin{align}
&\Pr\left[ \chi - \E\left[\chi \middle| \equal \right] < -\frac{1}{4}  \E\left[\chi \middle| \equal\right] \middle| \equal \right] \nonumber \\
&\quad \leq
\exp\left( -2\left( \frac{1}{m+1} + \frac{1}{n-m+1}\right) \left( \right( \E\left[\chi\middle|\equal\right] / 4 \left)^2 - 1 \right)\right) \nonumber \\
&\quad \leq \exp\left( -\frac{1}{15m} \left( \E\left[\chi\middle|\equal\right]\right)^2 \right) \leq \exp\left(- \frac{1}{15} \cdot \frac{N_{i-1}}{c_N}\right)
\end{align}
for sufficiently large $N$, where we use $c_N=o(N^{\frac{1}{2^l-1}})$ in evaluating $\alpha_{n_1,m,n}$.
Hence
\begin{equation}
\Pr\left[ \chi \geq N_{i-1} \middle| \equal \right] \geq 1 - \exp\left(- \frac{1}{15} \cdot \frac{N_{i-1}}{c_N}\right)
\end{equation}
holds, which implies that
\begin{align}
\Pr\left[\left| {\rm Im}(f) \cap L'_{i-1}\right| \geq N_{i-1} \middle| \pregood^{(i)} \right]
&=
\Pr\left[\chi \geq N_{i-1} \middle| \pregood^{(i)} \right] \nonumber \\
&\geq
\Pr\left[\chi \geq N_{i-1} \middle| \equal \right] \nonumber \\
&\geq 1 - \exp\left(- \frac{1}{15} \cdot \frac{N_{i-1}}{c_N}\right).
\end{align}

Therefore we have that
\begin{align}
\Pr\left[ \good^{(i)} \right]
&>
\Pr\left[ \good^{(i)} \middle| \pregood^{(i)} \right] \cdot \Pr\left[ \pregood^{(i)} \right] \nonumber\\
&= \Pr\left[ \left| {\rm Im}(f) \cap L'_{i-1}\right| \geq N_{i-1}  \middle| \pregood^{(i)} \right] \cdot \Pr\left[ \pregood^{(i)} \right] \nonumber\\
&\geq \left(1 - \frac{2}{|Y|} \right)\left( 1 - \exp\left(- \frac{1}{15} \cdot \frac{N_{i-1}}{c_N}\right) \right) \nonumber \\
&\geq 1 - \frac{2}{|Y|}  - \exp\left(- \frac{1}{15} \cdot \frac{N_{i-1}}{c_N}\right),
\end{align}
which completes the proof.
\qed
\end{proof}

Let $\good$ denote the event $\good^{(1)} \land \cdots \land \good^{(l)}$.
Then we can show the following claim.
\begin{claim}
For sufficiently large $N$, it holds that
\begin{equation}
\E\left[ Q \mymiddle \good \right] \leq \frac{1}{k}\mathsf{Qlimit}_k,
\end{equation}
where $Q$ is the total number of queries made by $\Mclaw_k$.
\end{claim}
\begin{proof}
Let us fix $i$ and $j$.
Let $Q^{(i)}_j$ denote the number of queries made by $\Mclaw_k$ in the $j$-th search to construct $i$-claws, and $Q^{(i)}$ denote $\sum_j Q^{(i)}_j$.
In the $j$-th search to construct $i$-claw, we search $X_i$ for $x$ with $f_i(x)\in L'_{i-1}$, where there exist at least $|L'_{i-1} \cap {\rm Im}(f)| \geq N_{i-1} -j+1$ answers in $X_i$ under the condition that $\good^{(i)}$ occurs.
From~\autoref{cor:MTPS}, the expected value of the number of queries made by $\MTPS$ in the $j$-th search to construct $i$-claws is upper bounded by
\begin{align}
9 \sqrt{5\card{X_i}/\card{f^{-1}_i(L'_{i-1})}} &\leq 9 \sqrt{5\card{X_i}/\card{L'_{i-1} \cap {\rm Im}(f)}} 
\leq 21 \sqrt{N/N_{i-1}}
\end{align}
for each $j$ under the condition that $\good^{(i)}$ occurs, for sufficiently large $N$ (we used the condition that $N_{i-1} = \omega(c_N N_i)$ holds for the last inequality).

Hence it follows that
\begin{align*}
 \E\left[Q^{(i)} \mymiddle \good^{(i)}\right]
 &= \E\left[\sum_{j} Q_j^{(i)} \mymiddle \good^{(i)}\right] = \sum_j \E\left[Q_j^{(i)} \mymiddle \good^{(i)}\right] \\
 &\leq \sum_{1 \leq j \leq \lceil 4c_N N_i \rceil} 21 \sqrt{{N}/{N_{i-1}}} 
 \leq
 \begin{cases}
 169 c^{3/2}_N N^{\frac{2^{l-1}-1}{2^l-1}} & (i=1)\\
 85 c_N N^{\frac{2^{l-1}-1}{2^l-1}} & (i\geq 2)
 \end{cases}
\end{align*}
for sufficiently large $N$.
Hence we have \hhnote{that $\E[Q \mid  \good]
 = \sum_{i} \E\left[Q^{(i)} \mymiddle \good^{(i)}\right]$ is upper bounded by} 
\begin{align*}
169c^{3/2}_N N^{\frac{2^{l-1}-1}{2^{l}-1}} 
     + \sum_{i=2}^{l} 85 c_N N^{\frac{2^{l-1}-1}{2^{l}-1}}
 &\leq 169 l c^{3/2}_N \cdot  N^{\frac{2^{l-1}-1}{2^{l}-1}}  = \frac{1}{k} \mathsf{Qlimit}_k,
\end{align*}
which completes the proof.
\qed
\end{proof}

From the above claims it follows that \hhnote{$\E[Q]$ is upper bounded by}
\begin{align}
\E[Q \mid  \good] + \E[Q \mid \lnot \good] \Pr[\lnot \good] 
&\leq \left( \frac{1}{k} + \Pr[\lnot \good]  \right) \cdot \mathsf{Qlimit}_k \label{eq:ExpectQupper},
\end{align}
and \hhnote{$\Pr\left[\lnot \good\right]$ is upper bounded by $\sum_{i} \Pr\left[\neg \good^{(i)}\right]$, which is further upper bounded by}
\begin{align}
\sum_{i} \left( \frac{2}{N}  + \exp\left(- \frac{1}{15} \cdot \frac{N_{i-1}}{c_N}\right) \right)
\leq \frac{2l}{N}  + l \cdot \exp\left(- \frac{1}{15} \cdot \frac{N^{\frac{1}{2^l - 1}}}{c_N}\right). \label{eq:GoodProbBound}
\end{align}

From Markov's inequality, the probability that $Q$ reaches $\mathsf{Qlimit}_k$ is at most
\begin{equation}
\Pr \left[ Q \geq \mathsf{Qlimit}_k \right] \leq \frac{\E[Q]}{\mathsf{Qlimit}_k} \leq \frac{1}{k} + \Pr[\lnot \good]. \label{eq:ProbReachLimit}
\end{equation}
The event ``$Q$ does not reach $\mathsf{Qlimit}_k$'' implies that $\Mclaw_k$ finds an $l$-claw.
Thus, from~\autoref{eq:GoodProbBound} and~\autoref{eq:ProbReachLimit}, the probability that $\Mclaw_k$ finds an $l$-claw is lower bounded by
\begin{equation}
1 - \frac{1}{k} - \frac{2l}{N}  - l \cdot \exp\left(- \frac{1}{15} \cdot \frac{N^{\frac{1}{2^l - 1}}}{c_N}\right),
\end{equation}
which completes the proof.
\qed
\end{proof}

\section{Conclusion}
\label{sec:conclusion}

This paper \tnote{has} developed a new quantum algorithm to find multicollisions of random functions.
Our new algorithm finds an $l$-collision of a random function $F \colon [N] \rightarrow [N]$ with $O\left(N^{(2^{l-1}-1)/(2^l-1)}\right)$ quantum queries on average, which improves the previous upper bound $O(N^{(3^{l-1}) / (2 \cdot 3^{l-1})})$ by Hosoyamada \etal~\cite{DBLP:conf/asiacrypt/HosoyamadaSX17}.
In fact, our algorithm can find \tnote{an} $l$-claw of random functions  $f_i \colon [N] \rightarrow [N]$ for $1 \leq i \leq l$ with the same average complexity $O\left(N^{(2^{l-1}-1)/(2^l-1)}\right)$.
In describing the algorithm,
we assumed for ease of analysis and understanding that intermediate measurements were allowed.
However, it is easy to move all measurements to the final step of the algorithm by the standard techniques.
In this paper, we focused only on query complexity, and did not provide the analyses of other complexity measures. However, it is not difficult to show that the space complexity and the depth of quantum circuits
are both bounded by $\tilde{O}\left(N^{(2^{l-1}-1)/(2^l-1)}\right)$.
For applications to cryptanalyses, it is of interest to further study time-and-memory-efficient variants.

\bibliographystyle{alpha}
\bibliography{ms}

\newcommand{\etalchar}[1]{$^{#1}$}
\begin{thebibliography}{BDH{\etalchar{+}}01}

\bibitem[Amb04]{DBLP:conf/focs/Ambainis04}
Andris Ambainis.
\newblock Quantum walk algorithm for element distinctness.
\newblock In {\em 45th Symposium on Foundations of Computer Science {(FOCS}
  2004), 17-19 October 2004, Rome, Italy, Proceedings}, pages 22--31, 2004.

\bibitem[BBHT98]{boyer1998tight}
Michel Boyer, Gilles Brassard, Peter H{\o}yer, and Alain Tapp.
\newblock Tight bounds on quantum searching.
\newblock {\em Fortschritte der Physik: Progress of Physics}, 46(4-5):493--505,
  1998.

\bibitem[BDH{\etalchar{+}}01]{DBLP:conf/coco/BuhrmanDHHMSW01}
Harry Buhrman, Christoph D{\"{u}}rr, Mark Heiligman, Peter H{\o}yer,
  Fr{\'{e}}d{\'{e}}ric Magniez, Miklos Santha, and Ronald de~Wolf.
\newblock Quantum algorithms for element distinctness.
\newblock In {\em Proceedings of the 16th Annual {IEEE} Conference on
  Computational Complexity, Chicago, Illinois, USA, June 18-21, 2001}, pages
  131--137, 2001.

\bibitem[BDRV18]{DBLP:conf/eurocrypt/BermanDRV18}
Itay Berman, Akshay Degwekar, Ron~D. Rothblum, and Prashant~Nalini Vasudevan.
\newblock Multi-collision resistant hash functions and their applications.
\newblock In {\em Advances in Cryptology - {EUROCRYPT} 2018 - 37th Annual
  International Conference on the Theory and Applications of Cryptographic
  Techniques, Tel Aviv, Israel, April 29 - May 3, 2018 Proceedings, Part {II}},
  pages 133--161, 2018.

\bibitem[Bel12]{DBLP:conf/focs/Belovs12}
Aleksandrs Belovs.
\newblock Learning-graph-based quantum algorithm for $k$-distinctness.
\newblock In {\em 53rd Annual {IEEE} Symposium on Foundations of Computer
  Science, {FOCS} 2012, New Brunswick, NJ, USA, October 20-23, 2012}, pages
  207--216, 2012.

\bibitem[BHT98]{DBLP:conf/latin/BrassardHT98}
Gilles Brassard, Peter H{\o}yer, and Alain Tapp.
\newblock Quantum cryptanalysis of hash and claw-free functions.
\newblock In {\em {LATIN} '98: Theoretical Informatics, Third Latin American
  Symposium, Campinas, Brazil, April, 20-24, 1998, Proceedings}, pages
  163--169, 1998.

\bibitem[BKP18]{DBLP:conf/stoc/BitanskyKP18}
Nir Bitansky, Yael~Tauman Kalai, and Omer Paneth.
\newblock Multi-collision resistance: a paradigm for keyless hash functions.
\newblock In {\em Proceedings of the 50th Annual {ACM} {SIGACT} Symposium on
  Theory of Computing, {STOC} 2018, Los Angeles, CA, USA, June 25-29, 2018},
  pages 671--684, 2018.

\bibitem[CNS17]{DBLP:conf/asiacrypt/ChaillouxNS17}
Andr{\'{e}} Chailloux, Mar{\'{\i}}a Naya{-}Plasencia, and Andr{\'{e}}
  Schrottenloher.
\newblock An efficient quantum collision search algorithm and implications on
  symmetric cryptography.
\newblock In {\em Advances in Cryptology - {ASIACRYPT} 2017 - 23rd
  International Conference on the Theory and Applications of Cryptology and
  Information Security, Hong Kong, China, December 3-7, 2017, Proceedings, Part
  {II}}, pages 211--240, 2017.

\bibitem[Gro96]{DBLP:conf/stoc/Grover96}
Lov~K. Grover.
\newblock A fast quantum mechanical algorithm for database search.
\newblock In {\em Proceedings of the Twenty-Eighth Annual {ACM} Symposium on
  the Theory of Computing, Philadelphia, Pennsylvania, USA, May 22-24, 1996},
  pages 212--219, 1996.

\bibitem[HS05]{hush2005concentration}
Don Hush and Clint Scovel.
\newblock Concentration of the hypergeometric distribution.
\newblock {\em Statistics \& probability letters}, 75(2):127--132, 2005.

\bibitem[HSX17]{DBLP:conf/asiacrypt/HosoyamadaSX17}
Akinori Hosoyamada, Yu~Sasaki, and Keita Xagawa.
\newblock Quantum multicollision-finding algorithm.
\newblock In {\em Advances in Cryptology - {ASIACRYPT} 2017 - 23rd
  International Conference on the Theory and Applications of Cryptology and
  Information Security, Hong Kong, China, December 3-7, 2017, Proceedings, Part
  {II}}, pages 179--210, 2017.

\bibitem[JLM14]{DBLP:conf/asiacrypt/JovanovicLM14}
Philipp Jovanovic, Atul Luykx, and Bart Mennink.
\newblock Beyond $2^{c/2}$ security in sponge-based authenticated encryption
  modes.
\newblock In {\em Advances in Cryptology - {ASIACRYPT} 2014 - 20th
  International Conference on the Theory and Application of Cryptology and
  Information Security, Kaoshiung, Taiwan, R.O.C., December 7-11, 2014.
  Proceedings, Part {I}}, pages 85--104, 2014.

\bibitem[KNY18]{DBLP:conf/eurocrypt/KomargodskiNY18}
Ilan Komargodski, Moni Naor, and Eylon Yogev.
\newblock Collision resistant hashing for paranoids: Dealing with multiple
  collisions.
\newblock In {\em Advances in Cryptology - {EUROCRYPT} 2018 - 37th Annual
  International Conference on the Theory and Applications of Cryptographic
  Techniques, Tel Aviv, Israel, April 29 - May 3, 2018 Proceedings, Part {II}},
  pages 162--194, 2018.

\bibitem[LZ18]{DBLP:journals/iacr/LiuZ18}
Qipeng Liu and Mark Zhandry.
\newblock On finding quantum multi-collisions.
\newblock {\em {IACR} Cryptology ePrint Archive}, 2018:1096, 2018.

\bibitem[MU17]{mitzenmacher2017probability}
Michael Mitzenmacher and Eli Upfal.
\newblock {\em Probability and computing: Randomization and probabilistic
  techniques in algorithms and data analysis}.
\newblock Cambridge university press, 2017.

\bibitem[RS96]{DBLP:conf/spw/RivestS96}
Ronald~L. Rivest and Adi Shamir.
\newblock Payword and micromint: Two simple micropayment schemes.
\newblock In {\em Security Protocols, International Workshop, Cambridge, United
  Kingdom, April 10-12, 1996, Proceedings}, pages 69--87, 1996.

\bibitem[Tan09]{DBLP:journals/tcs/Tani09}
Seiichiro Tani.
\newblock Claw finding algorithms using quantum walk.
\newblock {\em Theor. Comput. Sci.}, 410(50):5285--5297, 2009.

\bibitem[Zha15]{DBLP:journals/qic/Zhandry15}
Mark Zhandry.
\newblock A note on the quantum collision and set equality problems.
\newblock {\em Quantum Information {\&} Computation}, 15(7{\&}8):557--567,
  2015.

\end{thebibliography}

\end{document}